\documentclass[sn-mathphys-num]{sn-jnl}

\usepackage{graphicx}%
\usepackage{multirow}%
\usepackage{amsmath}%
\usepackage{amssymb}%
\usepackage{amsfonts}%
\usepackage{amsthm}%
\usepackage{mathrsfs}%
\usepackage[title]{appendix}%
\usepackage{xcolor}%
\usepackage{textcomp}%
\usepackage{manyfoot}%
\usepackage{booktabs}%
\usepackage{algorithm}%
\usepackage{algorithmicx}%
\usepackage{algpseudocode}%
\usepackage{listings}%
\usepackage{tikz}

\newtheorem{theorem}{Theorem}
\newtheorem{proposition}[theorem]{Proposition}%
\newtheorem{lemma}[theorem]{Lemma}

\newcommand{\mR}{\mathbf{R}}
\newcommand{\mRt}{\mathbf{R^\prime}}
\newcommand{\mQ}{\mathbf{Q}}
\newcommand{\mQt}{\mathbf{Q^\prime}}
\newcommand{\Ge}{\mathit{\Gamma}}
\newcommand{\Gae}{\mathit{\Gamma}^\prime}
\newcommand{\va}{\mathbf{a}}
\newcommand{\ue}{\mathrm{e}}
\newcommand{\ui}{\mathrm{i}}
\newcommand{\ud}{\mathrm{d}}
\newcommand{\vell}{\boldsymbol\ell}

\begin{document}

\title[Article Title]{Spectral Determinants of Almost Equilateral Quantum Graphs}


\author[1]{\fnm{Jonathan} \sur{Harrison}}

\author[2]{\fnm{Tracy} \sur{Weyand}}

\affil[1]{\orgdiv{Department of Mathematics}, \orgname{Baylor University}, \orgaddress{\street{1410 S. 4th Street}, \city{Waco}, \state{Texas}, \postcode{76706}, \country{United States}}, ORCHID ID 0000-0002-2590-4503}

\affil[2]{\orgdiv{Department of Mathematics}, \orgname{Rose-Hulman Institute of Technology}, \orgaddress{\street{5500 Wabash Avenue}, \city{Terre Haute}, \state{Indiana}, \postcode{47803},  \country{United States}}, ORCHID ID 0000-0001-9735-5825}


\abstract{Kirchoff's matrix tree theorem of 1847 connects the number of spanning trees of a graph to the spectral determinant of the discrete Laplacian \cite{K47}.   Recently an analogue was obtained for quantum graphs  relating the number of spanning trees to the spectral determinant of a Laplacian acting on functions on a metric graph with standard (Neumann-like)
vertex conditions \cite{HWST}.  
	This result holds for quantum graphs where the edge lengths are close together.  	A quantum graph where the edge lengths are all equal is called equilateral. 
	Here we consider equilateral graphs where we perturb the length of a single edge (almost equilateral graphs).
 We analyze the spectral determinant of almost equilateral complete graphs, complete bipartite graphs, and circulant graphs.   
	 This provides a measure of how fast the spectral determinant changes with respect to changes in an edge length.  We apply these results to estimate the width of a window of edge lengths where the connection between the  number of spanning trees and the spectral determinant can be observed.  The results suggest the connection holds for a much wider window of edge lengths than is required in \cite{HWST}. 
}

\keywords{spectral graph theory, quantum graphs, spectral determinant, spanning trees}

\pacs[MSC Classification]{35A01, 65L10, 65L12, 65L20, 65L70}

\maketitle

\section{Introduction}\label{sec1}

A classical result in spectral graph theory is Kirchhoff's matrix tree theorem \cite{K47}.  The number of spanning trees of a graph is given by the spectral determinant; the product of the non-zero eigenvalues of the graph's discrete Laplacian.  

\begin{theorem}[Kirchhoff's Matrix Tree Theorem]\label{thm KMT}
For a connected graph $G$ with $V$ vertices,
\begin{equation}
\# \textrm{ spanning trees } = \frac{1}{V} \det{'} (\mathbf{L}) =\det(\mathbf{L}[i]) \ ,
\end{equation}
 where $i=1,\dots,V$.  The $V\times V$ matrix $\mathbf{L}=\mathbf{D}-\mathbf{A}$ is the discrete Laplacian with $\mathbf{D}$ the diagonal matrix of vertex degrees and $\mathbf{A}$ the adjacency matrix of the graph.  The $(V-1)\times (V-1)$ matrix $\mathbf{L}[i]$ is produced by deleting the $i$th row and column from $\mathbf{L}$.
\end{theorem}

In \cite{HWST} the authors obtained a similar result for quantum graphs.  A quantum graph is a metric graph along with a self-adjoint operator acting on functions defined on a set of intervals associated with the edges of the graph. Functions on these intervals are related by local boundary conditions at the vertices; see e.g. \cite{Kuc03,BKbook} for an introduction to quantum graphs.  The spectral determinant of the quantum graph Laplacian with standard (Neumann-like) vertex conditions determines the number of spanning trees when the edge lengths are close together; see \cite{HWST} theorem 2, which is reproduced below.
\begin{theorem}\label{thm:generic spanning trees}
	Let $\Gamma$ be a connected metric graph with $V$ vertices and $E$ edges whose edge lengths are in the interval $[\ell,\ell+\delta)$.  If
	\begin{equation}\label{eq: delta constraint}
	\delta< \frac{\ell}{V^V \, 2^{E+V}\sqrt{2EV}}  \ ,
	\end{equation}
	then the number of spanning trees is the closest integer to
	\begin{equation}\label{eq:T thm}
	T_{\Gamma}=  \dfrac{\prod_{v\in\mathcal{V}} d_v}{E \, 2^E \,\ell^{\beta+1}} {\det}'(\mathcal{L}) \ ,
	\end{equation}
	where $\mathcal{V}$ is the set of vertices, $d_v$ is the degree of vertex $v$, $\mathcal{L}$ is the Laplacian with standard vertex conditions, and $\beta=E-V+1$ is the first Betti number of $\Gamma$, the number of independent cycles.
\end{theorem}
The spectral determinant can also be used to compute magnetization and transport properties of the network \cite{PM97,PM99}.
It was formulated for Laplace and Schr\"odinger operators on metric graphs in \cite{D00,D01,Fr}, via periodic orbits or pseudo-orbits in \cite{Aetal00, BHJ12, WGTR}, and using the zeta function in \cite{HK11,HKT12,HWK16}. 

To obtain theorem \ref{thm:generic spanning trees}, the edge lengths of the metric graph are tightly constrained so they are all close to $\ell$.  We call a metric graph where the edge lengths are all equal an equilateral graph.  Hence the theorem \ref{thm:generic spanning trees} holds for metric graphs that are close to an equilateral graph.   However, heuristically, the spectral determinant of a quantum graph should determine the number of spanning trees for graphs where the edge lengths are less tightly constrained.  

For instance, when computing the spectral statistics of large quantum graphs (seen in applications to quantum chaos), a simplification
that is frequently used \cite{SS00,T01,BSW02,BSW03}
is to instead compute the spectral statistics of the eigenphases of the unitary scattering matrix of the graph on the unit circle.
In both cases the spectrum is unfolded to have a mean spacing of one.  
The unfolded spectrum of an equilateral graph is the periodic extension of the unfolded spectrum of the eigenphases of the scattering matrix.  The explanation for this correspondence is that roots of the secular equation,
\begin{equation}
	\det (\mathbf{I}-\mathbf{S}\ue^{\ui k \mathbf{M}}) = 0,
\end{equation}  
are the square roots of the eigenvalues of the Laplace operator on the metric graph, where $\mathbf{S}$ is the unitary graph scattering matrix and $\mathbf{M}$ is the corresponding diagonal matrix of edge lengths.   Changing the edge lengths without changing the total length of the graph jiggles the roots locally without changing the mean spacing, level repulsion or statistical properties.  Consequently, we hypothesize that, the spectral determinant should be somewhat robust to such changes.   This would also imply that the spanning tree formula (\ref{eq:T thm}) is robust and could be accurate for graphs where the mean spacing of edge lengths is $\ell$ even if there is variation in the individual edge lengths. 

The connection between eigenphase statistics and spectral statistics was made precise in \cite{BW10} for the nearest neighbor spacing statistics when the set of edge lengths of the quantum graph are incommensurate.  If the edge lengths are rationally related, for example in an equilateral graph, the nearest neighbor spacings in the graph spectrum are highly degenerate.  The relationship between the two forms of graph quantization, a self-adjoint operator on a network of intervals versus a collection of unitary vertex scattering matrices, was discussed in \cite{B08,H24}.

In this article we approach the relationship between the spectral determinant of a quantum graph and the number of spanning trees from a different direction.  We consider three classes of equilateral graphs where we can compute the spectral determinant analytically; complete graphs, complete bipartite graphs, and connected circulant graphs of prime order where the maximum multiplicity of eigenvalues of the adjacency matrix is two (which is the generic case). We then consider perturbing the length of a single edge to obtain an almost equilateral graph, and we ascertain how the spectral determinant and the number of spanning trees changes under the perturbation.  This provides a measure of how fast $T_\Gamma$ varies as one changes an edge length.  In each case, we see that $T_\Gamma$ can be expected to provide an accurate determination of the number of spanning trees for graphs where the edge lengths vary much more widely than is implied by the constraint (\ref{eq: delta constraint}) on $\delta$ imposed in theorem \ref{thm:generic spanning trees}.  For example, the bound $\delta< \ell / V^V$ can be expected to provide a sufficiently narrow window of edge lengths for $T_{\Gamma}$ to determine the number of spanning trees for these classes of graphs.
Perturbations of graph operators have been considered in a number of settings, for example graph surgery \cite{BKKM}, the effect of shrinking an edge length of a quantum graph to zero in \cite{BLS}, and the effect of changing edge lengths on the spectral gap \cite{BL}.  

The article is laid out as follows.  In section \ref{sec: back} we introduce notation for graphs and quantum graphs.  Sections \ref{sec: transitive graphs} and \ref{sec: circulant} analyze the spectral determinants of almost equilateral complete graphs, complete bipartite graphs, and certain connected circulant graphs.  In section \ref{sec: spanning trees} we apply the results for the spectral determinants to see how the function $T_{\Gamma}$ behaves under perturbation of an edge length. For completeness, the appendices describe how individual eigenvalues of almost equilateral complete and complete bipartite graphs change under perturbation of an edge length.   

\section{Background}\label{sec: back}

A \textit{graph} $G$ consists of a set of vertices $\mathcal{V}$ and a set of edges $\mathcal{E}$ such that each edge $e = (u,v) \in \mathcal{E}$ connects a pair of vertices $u, v \in \mathcal{V}$; see  figures \ref{fig: complete bipartite} and \ref{fig: circulant} for examples. Two vertices are \textit{adjacent} if there is an edge $(u,v) \in \mathcal{E}$, denoted $u\sim v$. The number of vertices is $V = |\mathcal{V}|<\infty$, and $E = |\mathcal{E}|<\infty$ is the number of edges.  The \textit{degree} of vertex $v$ is the number of vertices adjacent to $v$, denoted $d_v$.  A graph is $k$-\emph{regular} if $d_v=k$ for all $v\in \mathcal{V}$.
Throughout the paper, we assume that graphs are \textit{connected}, so there is a path between every pair of vertices. We also assume the graphs are \emph{simple}, with no loops or multiple edges.  
The \textit{first Betti number} is $\beta = E-V+1$, the number of independent cycles of $G$. A \textit{tree} is a connected graph with no cycles ($\beta = 0$). A \textit{spanning tree} of $G$ is a subgraph that is a tree and contains all vertices of $G$.  
Additionally, when we reference sets of vectors and specific elements of vectors, we use the convention that $\mathbf{v}_j$ denotes the $j$th vector in the set $\{\mathbf{v}_j\}_{j=1}^n$ whereas $[\mathbf{v}]_j$ denotes the $j$th element of the vector $\mathbf{v}$. 

We now define some classes of graphs that will be used in the article.  The \emph{fully connected graph} on $n$ vertices $K_n$ is the graph where every pair of vertices is connected by an edge so $(u,v)\in \mathcal{E}$ for all $u,v\in \mathcal{V}$ with $u\neq v$.  The \emph{complete bipartite graph} $K_{m,n}$ is a graph where the set of vertices is divided into two disjoint sets, $\mathcal{V}=\mathcal{A}\cup \mathcal{B}$ with $\mathcal{A}\cap \mathcal{B}=\emptyset$.  Every vertex of $\mathcal{A}$ is adjacent to every vertex of $\mathcal{B}$ and there are no edges connecting pairs of vertices in $\mathcal{A}$ or pairs of vertices in $\mathcal{B}$.  The number of vertices in $\mathcal{A}$ is $m=|\mathcal{A}|$ and $n=|\mathcal{B}|$.  Figure \ref{fig: complete bipartite} shows the complete bipartite graph $K_{3,5}$.
A \emph{circulant graph} on $n$ vertices is denoted $C_n(\va )$ where $\va=(a_1,\dots a_p)$ with $a_1 < a_2 < \dots <a_p \leq n/2$.  Then two vertices $i,j\in \{ 0,\dots, n-1 \}$ are connected if $|i-j|=\pm a_k (\textrm{mod } n)$ for some $k\in \{1,\dots,p\}$.  If $a_p<n/2$, then the circulant graph is $2p$-regular.  A circulant graph is connected if and only if $\mathrm{gcd} (n,a_1,\dots ,a_p)=1$; see \cite{D86} corollary 1. Figure \ref{fig: circulant} shows the connected circulant graph $C_{17}(2,5)$.

\begin{figure}[tbh]
	\centering
	\includegraphics[width=5cm]{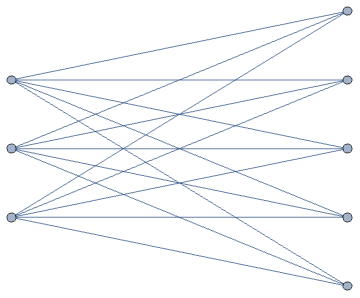}
\caption{\small{The complete bipartite graph $K_{3,5}$.}}\label{fig: complete bipartite}
\end{figure} 

\begin{figure}[tbh]
	\centering
	\includegraphics[width=5cm]{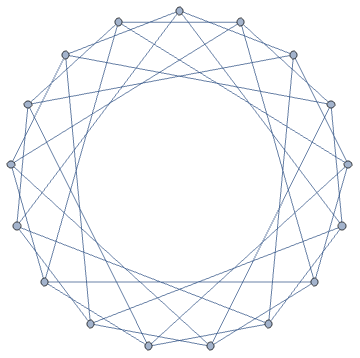}
	\caption{\small{The circulant graph $C_{17}(2,5)$.}}\label{fig: circulant}
\end{figure}

\subsection{Quantum Graphs}

A \textit{metric graph} $\Gamma$ is a graph where each edge $e\in \mathcal{E}$ is assigned a length $\ell_e > 0$ and we define a vector of edge lengths $\vell=(\ell_1,\dots,\ell_E)$ for some ordering of the edges of the graph.  The \textit{total length} of a metric graph is $\ell_{tot}=\sum_{e\in \mathcal{E}} \ell_e$.  We view the edge $e = (u,v) \in \mathcal{E}$ as the interval $[0, \ell_e]$ with $x_e = 0$ at $u$ and $x_e = \ell_e$ at $v$; the choice of orientation will not affect our results. This allows us to define a function $f$ on $\Gamma$ as a collection of functions, $\{f_e\}_{e\in\mathcal{E}}$, such that the function $f_e$ is defined on the interval $[0, \ell_e]$. A metric graph is \textit{equilateral} if all the edges have the same length, $\vell=(\ell,\dots,\ell)$. We use $\Ge$ to denote an equilateral graph.  If $\Ge$ is equilateral with edge lengths $\ell$, then we will say that $\Gae$ is an \textit{almost equilateral graph} if the length of one edge $e$ of $\Ge$ is perturbed so $\ell_e=\ell + \epsilon$. 

A \textit{quantum graph} is a metric graph $\Gamma$ along with a self-adjoint differential operator. Our focus in this paper is on the \textit{Laplace operator}, which acts as $-\frac{\mathrm{d}^2}{\mathrm{d}x_e^2}$ on functions that are defined on the interval $[0, \ell_e]$. We consider standard Neumann-Kirchhoff conditions at the vertices:
\begin{equation}\label{eq: vertex conditions}
\left\{\begin{array}{l}
f\mbox{ is continuous on } \Gamma \mbox{ and}\\
\sum_{e\in \mathcal{E}_v} \frac{\mathrm{d}f_e}{\mathrm{d}x_e}(v) = 0 \mbox{ at each vertex } v,
\end{array}\right.
\end{equation}
where $\mathcal{E}_v$ denotes the edges containing the vertex $v$ and the derivative $\frac{\mathrm{d}f_e}{\mathrm{d}x_e}(v)$ is taken in the outgoing direction at $v$. The second Sobolev space on $\Gamma$ is the direct sum of second Sobolev spaces on the set of intervals:
\begin{equation}
H^2(\Gamma) = \bigoplus_{e\in\mathcal{E}} H^2([0, \ell_e]) \ .
\end{equation}
 The domain of the Laplace operator is the set of functions $f \in H^2(\Gamma)$ that satisfy \eqref{eq: vertex conditions}. With these conditions, the Laplace operator is self-adjoint with infinitely many non-negative eigenvalues; see e.g. \cite{KS99}. We use the notation $\mathcal{L}$ and $\mathcal{L^\prime}$ for the Laplace operator on an equilateral graph $\Ge$ and a corresponding almost equilateral graph $\Gae$ respectively.

The spectral determinant of a quantum graph is formally the product of the non-zero eigenvalues of the graph Laplacian,
\begin{equation}
\det{'}(\mathcal{L})(\vell) = \prod_{j=0}^\infty{'} \lambda_j \ ,
\end{equation}
where the prime indicates that any eigenvalues of zero are omitted from the product.  A regularized spectral determinant can be obtained from the spectral zeta function,
\begin{equation}
\zeta (s) = \sum_{j=0}^\infty{'} \lambda_j^{-s} \ ,
\end{equation}
via $\det{'}(\mathcal{L})(\vell) = \exp \big( -\zeta'(0) \big)$.   Then, using the Dirichlet-to-Neumann operator, Friedlander \cite{Fr} formulates the spectral determinant in the following theorem.
\begin{theorem}\label{thm: Fr}
	Given a connected metric graph,
	\begin{equation}\label{eq: spec det formula}
	\det{'}(\mathcal{L})(\vell) = \frac{2^E}{V}\frac{\ell_{tot}\prod_{e\in\mathcal{E}}\ell_e}{\prod_{v\in\mathcal{V}} d_v}\det{'}(\mR),
	\end{equation}
	where $\mR$ is the $V \times V$ matrix,
	\begin{equation}
	[\mR]_{u,v} = \begin{cases}
	\sum_{w\sim v} \ell_{(w,v)}^{-1}, & \mbox{ if } u=v\\
	-\ell_{(u,v)}^{-1}, & \mbox{ if } u\sim v\\
	0, & \mbox{ otherwise}.
	\end{cases}
	\end{equation}
\end{theorem}
The spectral determinant can also be obtained from the zeta function using a contour integral approach \cite{HK11,HKT12} .  Alternatively, writing the spectrum as a sum over periodic orbits via the trace formula leads to periodic orbit expressions for the spectral determinant \cite{Aetal00, BHJ12, WGTR}, but it is the formulation provided in theorem \ref{thm: Fr} that will prove useful here.

\section{Almost Equilateral Edge Transitive Graphs}
\label{sec: transitive graphs}

In this section we evaluate the spectral determinant of two families of almost equilateral edge transitive graphs; complete graphs and complete bipartite graphs. A graph is \textit{edge transitive} if given any two edges $e_i$ and $e_j$, there is a graph automorphism that maps $e_i$ to $e_j$. The spectral determinant of an almost equilateral edge transitive graph can be obtained directly from the spectral determinant of the corresponding equilateral graph due to the high degree of symmetry in an edge transitive graph.


To evaluate the spectral determinant of the equilateral graph we use theorem \ref{thm: Fr}.
In an edge transitive graph each edge is indistinguishable from any other edge when looking at the connections around it. Consequently for all $i$ and $j$,
\begin{equation}
\left. \frac{\partial}{\partial \ell_i} \det{'}(\mathcal{L}) \right|_{\vell=(\ell,\dots,\ell)} 
= \left. \frac{\partial}{\partial \ell_j} \det{'}(\mathcal{L}) \right|_{\vell=(\ell,\dots,\ell)} \ .
\end{equation}
Hence, by the chain rule, 
\begin{equation}\label{eq: edge transitive}
\frac{\ud}{\ud \ell} \left( \det{'} (\mathcal{L}) (\ell,\dots,\ell) \right) = E \left. \frac{\partial}{\partial \ell_1} \det{'}(\mathcal{L}) \right|_{\vell=(\ell,\dots,\ell)} \ ,
\end{equation}
where we will assume, without loss of generality, that the perturbed edge of the almost equilateral graph is the first edge. 


For example, two families of edge transitive graphs are complete graphs and complete bipartite graphs. 
Using theorem \ref{thm: Fr}, the spectral determinant of an equilateral complete graph $K_n$ is,
\begin{equation}
\det{'} (\mathcal{L}) (\ell,\dots,\ell) = \frac{E2^{E}n^{n-2}\ell^{\beta+1}}{(n-1)^n} \ .
\end{equation} 
Therefore, by equation \eqref{eq: edge transitive}, the spectral determinant of the almost equilateral graph where one edge has length $\ell+\epsilon$ is, 
 \begin{align}
  \det{'}(\mathcal{L}^\prime) 
 &=  \frac{E2^{E}n^{n-2}\ell^{\beta+1}}{(n-1)^n} \left(  1+\frac{\epsilon(\beta+1)}{\ell E} \right) +\mathcal{O}(\epsilon^2) \ , \label{eq: C SD}
  \end{align}
where $E=n(n-1)/2$ and $\beta=E-n+1$.


Similarly, the spectral determinant of an equilateral complete bipartite graph $K_{n,m}$ is, 
\begin{equation}
\det{'}(\mathcal{L})  (\ell,\dots,\ell) = 2^E \ell^{\beta+1} \ .
\end{equation}
Hence, applying (\ref{eq: edge transitive}), the spectral determinant of the almost equilateral complete bipartite graph is,
 \begin{align}
\det{'}(\mathcal{L}^\prime) 
&=  2^E \ell^{\beta+1} \left(  1+\frac{\epsilon(\beta+1)}{\ell E} \right) +\mathcal{O}(\epsilon^2) \ , \label{eq: CB SD}
\end{align}
where $E=mn$ and $\beta=E-(m+n)+1$.

\section{Almost Equilateral Circulant Graphs}\label{sec: circulant}

{Consider a connected $2p$-regular circulant graph $C_n(\va)$ with 
 $\va = (a_1, a_2, \ldots, a_p)$ and $a_1 < a_2 < \ldots < a_p < n/2$; see e.g. figure \ref{fig: circulant}.  
 We first note that, unlike complete graphs and complete bipartite graphs, a circulant graph is not edge transitive.  An edge of a circulant graphs whose vertices differ by $a_i$ modulo $n$ is mapped by a graph automorphism to another edge where the vertices also differ by $a_i$.  Consequently, there are $p$ independent partial derivatives of the spectral determinant of an equilateral circulant graph.
 To formulate the spectral determinant we first determine how individual eigenvalues of $\mR$ change under perturbation of an edge.  This approach can also be applied to the complete and complete bipartite graphs as described in the appendices.}

Let $\Ge$ be an equilateral metric graph. In this case, the matrix $\mR = \ell^{-1}\mathbf{L}$ in theorem \ref{thm: Fr} is
\begin{equation}\label{eq: R}
[\mR]_{u,v} = \begin{cases}
\ell^{-1}d_v,  & \mbox{if  } u=v\\
-\ell^{-1},  & \mbox{if  } u \sim v\\
0, &\mbox{otherwise},
\end{cases}
\end{equation}
where $\mathbf{L}$ is the \textit{combinatorial Laplacian matrix}; $\mathbf{L}=\mathbf{D}-\mathbf{A}$ where $\mathbf{D}$ is the diagonal matrix of vertex degrees with $[\mathbf{D}]_{v,v} = d_v$ and $\mathbf{A}$ is the adjacency matrix with $[\mathbf{A}]_{u,v} = 1$ if $u\sim v$ and $0$ otherwise.  Note that, as the edges are not directed, the adjacency matrix is by definition symmetric. 

The adjacency matrix $\mathbf{A}$ of a circulant graph is a circulant matrix, and consequently the matrices $\mathbf{L}$ and $\mR$ are also circulant matrices which share the same eigenvectors.  To use perturbation theory we want to know the multiplicity of the eigenvalues of $\mR$.  We will see that most eigenvalues of $\mR$ have a multiplicity of at least two.  We will fix $n$ to be prime, which guarantees that the graph is connected, and choose $\va= (a_1,  \ldots, a_p)$ so there are no eigenvalues of $\mathbf{A}$ (or equivalently $\mR$) with multiplicity greater than two, which is the generic case.  For example, if $p=2$ the maximum multiplicity is two provided  $n \nmid (a_1^2+a_2^2)$ which is satisfied if $n$ is large enough  \cite{HP24}.  The restriction to circulant graphs where the maximum multiplicity of eigenvalues of the adjacency matrix is two is convenient for the perturbation theory calculation.  However, if there were eigenvalues of $\mR$ with higher multiplicity, the calculation can be done in the same way and we would expect similar conclusions.

\subsection{Perturbation Theory}\label{subsec: Perturbation}

We now form a perturbed graph, $\Gae$, by changing the length of one edge $e = (a,b) \in \mathcal{E}$ to $\ell + \epsilon$.  Then,
\begin{equation}\label{eq: tR}
[\mRt]_{u,v} = \begin{cases}
 \ell^{-1}d_v,  & \mbox{if  } u=v\neq a\mbox{ or } b\\
-\ell^{-1},  & \mbox{if  } (u,v) \in \mathcal{E} - (a,b)  \\
\ell^{-1}d_v - \ell^{-2}\epsilon + \mathcal{O}(\epsilon^2), &\mbox{if  } u=v = a\mbox{ or } b\\
-\ell^{-1} + \ell^{-2}\epsilon + \mathcal{O}(\epsilon^2), & \mbox{if } (u,v)=(a,b)\\
0, &\mbox{otherwise}.
\end{cases}
\end{equation}
Equivalently,
\begin{equation}
\mRt = \mR + \mathbf{M} + \mathcal{O}(\mathbf{M}^2) = \mathbf{R} - \frac{\epsilon}{\ell^2}\mQ + \mathcal{O}(\epsilon^2) \ ,
\end{equation}
where $\mathbf{M} = -\frac{\epsilon}{\ell^2}\mQ$ and 
\begin{equation}\label{eq: Q}
[\mQ]_{u,v} = \begin{cases}
1, & u = v = a\mbox{ or } b\\
-1, & \mbox{if } (u,v)=(a,b)\\
0, &\mbox{otherwise}.
\end{cases}
\end{equation}

We use  the following proposition to relate the non-degenerate eigenvalues of $\mR$ and $\mRt$; see e.g. \cite{MMP}.

\begin{proposition}\label{prop: nonde}
If the matrix $\mathbf{R}$ has a non-degenerate eigenvalue $\lambda$ with normalized eigenvector $\mathbf{v}$, then $\mRt$ has a non-degenerate  eigenvalue 
\begin{equation}
\lambda^\prime = \lambda -\frac{\epsilon}{\ell^2} \mathbf{v}\cdot \mathbf{Q}\mathbf{v} + \mathcal{O}(\epsilon^2) \ .
\end{equation} 
\end{proposition}

For example, assume that all the eigenvalues of the matrix $\mR$ are simple. Let $\{\mathbf{v}_k\}_{k=1}^V$ be the set of orthonormal eigenvectors corresponding to the eigenvalues $\{\lambda_k\}_{k=1}^V$. Using \eqref{eq: Q}, one can see that 
\begin{equation}
\mathbf{v}_k\cdot \mQ\mathbf{v}_k = ([\mathbf{v}_k]_a - [\mathbf{v}_k]_b)^2,
\end{equation}
and hence, by proposition \ref{prop: nonde}, the eigenvalues of $\mRt$ are
\begin{equation}
\lambda_k^\prime = \lambda_k - \frac{\epsilon}{\ell^2}([\mathbf{v}_k]_a-[\mathbf{v}_k]_b)^2 + \mathcal{O}(\epsilon^2) \ .
\end{equation}

We use the following proposition if an eigenvalue $\lambda$ is degenerate; see e.g. \cite{MMP}.

\begin{proposition}\label{prop: degen}
Suppose that $\lambda$ is an eigenvalue of $\mR$ with multiplicity $n$ and corresponding orthonormal eigenvectors $\{\mathbf{u}_k\}_{k=1}^n$. Let $\mathbf{Q^\prime}$ be the $n\times n$ matrix defined by,
\begin{equation}\label{eq: new Q}
[\mathbf{Q^\prime}]_{j,k} = \mathbf{u}_j\cdot \mathbf{Q}\mathbf{u}_k \ .
\end{equation}
Then $n$ of the eigenvalues of $\mRt$ are
\begin{equation}
\lambda_k^\prime = \lambda -\frac{\epsilon}{\ell^2} \lambda_k + \mathcal{O}(\epsilon^2) \ ,
\end{equation}
where $\lambda_k$ is the $k$th eigenvalue of $\mQt$. 
\end{proposition}

\subsection{The Matrix $\mR$ for a Circulant Graph}

Consider a $2p$-regular circulant graph $C_n(\va)$ with 
 $\va = (a_1, a_2, \ldots, a_p)$ and $a_1 < a_2 < \ldots < a_p < n/2$.  
 For an equilateral circulant graph, the matrix $\mR$ is a circulant matrix with entries
\begin{equation}
[\mR]_{u,v} = \begin{cases}
2p\ell^{-1}, & \mbox{if } u=v  \\
-\ell^{-1}, & \mbox{if }|u-v|= a_k (\mathrm{mod }\, n) \\
0, & \mbox{otherwise}  .
\end{cases}
\end{equation}
Eigenvalues of an $n\times n$ circulant matrix with first row,
\begin{equation}
\begin{bmatrix}
c_0 & c_1 & c_2 & \ldots & c_{n-1}
\end{bmatrix}
\end{equation}
are
\begin{equation}
\lambda_j = c_0 + c_1\omega^j + c_2\omega^{2j} + \ldots + c_{n-1}\omega^{(n-1)j} \ ,
\end{equation}
 with corresponding eigenvectors,
\begin{equation}\label{eq: eigenvector}
\mathbf{v}_j = \frac{1}{\sqrt{n}}\begin{bmatrix}
1, \omega^j, \omega^{2j}, \ldots, \omega^{(n-1)j}
\end{bmatrix}^T
\end{equation}
for $j = 0, 1, \ldots n-1$, where $\omega=\ue^{\frac{2\pi \ui}{n}}$ is a primitive $n$th root of unity \cite{Davis}. As $n$ is prime, $\omega^j$ is a primitive root of unity for $j=1,\dots, n-1$.
Therefore the eigenvalues of $\mR$ are $0$ (with multiplicity $1$) and $\lambda_j = 2p\ell^{-1} - 2\ell^{-1}\sum_{k=1}^p \cos\left(\frac{2\pi j a_k}{n}\right)$.  We notice that $\lambda_j=\lambda_{n-j}$ for $0<j<(n-1)/2$, so these eigenvalue have a multiplicity of at least two. Hence,
\begin{equation}\label{eq:detR circulant}
\det{'}(\mR) = \left( \frac{2}{\ell} \right)^{n-1} \prod_{j=1}^{(n-1)/2} \left( p-\sum_{k=1}^p \cos\left(\frac{2\pi j a_k}{n}\right)  \right)^2 \ .
\end{equation}
We assume that the vector $\va$ is chosen so the multiplicity of the non-zero eigenvalues is exactly two, which is the generic situation. 
We now consider how the eigenvalues change in an almost equilateral graph.

\subsection{Eigenvalues of $\mR'$}

Zero is the only non-degenerate eigenvalue of $\mR$ with a normalized eigenvector,
\begin{equation}
\mathbf{v}_0 = \frac{1}{\sqrt{n}}\begin{bmatrix}
1, \ldots, 1
\end{bmatrix}^T.
\end{equation} Then, from proposition \ref{prop: nonde}, the corresponding eigenvalue of $\mRt$ is,
\begin{equation}
\lambda_0' = 0 - \frac{\epsilon}{\ell^2} \mathbf{v}_0 \cdot\mQ\mathbf{v}_0 + \mathcal{O}(\epsilon^2) = 0 + \mathcal{O}(\epsilon^2) 
\end{equation}
since $\mQ\mathbf{v}_0 = \mathbf{0}$.

The remaining eigenvalues of $\mR$ have multiplicity two, so $\lambda_j=\lambda_{n-j}$ for $j=1,\dots, (n-1)/2$ with corresponding eigenvectors,
\begin{align}
\mathbf{v}_j &= \frac{1}{\sqrt{n}}\begin{bmatrix}
1, \omega^j, \omega^{2j}, \ldots, \omega^{(n-1)j}
\end{bmatrix}^T  \ ,\\
\mathbf{v}_{n-j} & = \frac{1}{\sqrt{n}}\begin{bmatrix}
1, \omega^{-j}, \omega^{-2j}, \ldots, \omega^{-(n-1)j}
\end{bmatrix}^T\ .
\end{align}
Without loss of generality, we assume the edge $e=(0,a_m)$ has its length perturbed. 
Then, 
\begin{equation}
[\mQ\mathbf{v}_j]_k = \frac{1}{\sqrt{n}}\begin{cases}
1-\omega^{a_mj}, & \mbox{if } k=0\\
\omega^{a_mj}-1, & \mbox{if } k=a_m\\
0, & \mbox{ otherwise}.
\end{cases} 
\end{equation} Therefore,
\begin{align}
\mathbf{v}_j\cdot \mQ\mathbf{v}_j  &= \frac{2}{n}\omega^{a_mj}\left(\cos\alpha_{m,j} - 1\right)\ , \\
\mathbf{v}_{n-j}\cdot \mQ\mathbf{v}_j & =- \frac{2}{n}\left(\cos\alpha_{m,j}-1\right)
\end{align} 
where 
\begin{equation}
\alpha_{m,j} = \frac{2\pi a_m j}{n} \ .
\end{equation}
Similarly,
\begin{align}
\mathbf{v}_{n-j}\cdot\mQ\mathbf{v}_{n-j} & = \frac{2}{n}\omega^{-a_mj}\left(\cos\alpha_{m,j} - 1\right) \ , \\
\mathbf{v}_j\cdot\mQ\mathbf{v}_{n-j} &=- \frac{2}{n}\left(\cos\alpha_{m,j} -1\right)
\end{align}
since $\cos\alpha_{m,n-j} = \cos\alpha_{m,j}$. Hence,
\begin{equation}
\mQt  = \frac{2\left(\cos\alpha_{m,j} - 1\right)}{n}\begin{bmatrix}
\omega^{a_mj} & -1\\
-1 & \omega^{-a_mj}
\end{bmatrix} \ .
\end{equation}
 The eigenvalues of the matrix in square brackets are $0$ and $2\cos\alpha_{m,j}$. So, by proposition \ref{prop: degen}, the eigenvalues of $\mRt$ are
\begin{align}
\lambda_j^\prime &= \lambda_j + \mathcal{O}(\epsilon^2) \ ,\\
 \lambda_{n-j}^\prime &= \lambda_j + \frac{\epsilon}{\ell^2}\frac{4\cos\alpha_{m,j}\left(\cos\alpha_{m,j} - 1\right)}{n} + \mathcal{O}(\epsilon^2)
\end{align}
for $j = 1, \ldots, (n-1)/2$.

\subsection{The Spectral Determinant}
With $n$ prime and $\va$ chosen so that the eigenvalues of $\mR$ have at most multiplicity two, the spectral determinant of $\mRt$ is,
\begin{align}
\det{'}(\mRt) 
& = \det{'}(\mR) + \frac{\epsilon}{\ell^2}\sum_{j=1}^{(n-1)/2} \frac{4\cos\alpha_{m,j}\left(\cos\alpha_{m,j} - 1\right)}{n}\frac{\det^\prime(\mR)}{\lambda_j}+\mathcal{O}(\epsilon^2)  \\
&= \det{'}(\mR) \left(1 + \frac{\epsilon}{\ell}\sum_{j=1}^{(n-1)/2} \frac{2\cos\alpha_{m,j}\left(\cos\alpha_{m,j} - 1\right)}{n(p-\sum_{k=1}^p\cos\alpha_{k,j})}\right) + \mathcal{O}(\epsilon^2) \label{eq:R' aecg}\ .
\end{align}
Hence, applying theorem \ref{thm: Fr}, the almost equilateral spectral determinant is,
\begin{equation}
\det{'}(\mathcal{L}{'}) =  \det{'}(\mathcal{L}) \left( 1+ \frac{\epsilon}{\ell}   
\left(  1+\frac{1}{np}+\sum_{j=1}^{(n-1)/2} \frac{2\cos\alpha_{m,j}\left(\cos\alpha_{m,j} - 1\right)}{n(p-\sum_{k=1}^p\cos\alpha_{k,j})} \right)\right) +\mathcal{O}(\epsilon^2)  \ , 
\end{equation}
where $(0,a_m)$ is the perturbed edge and 
\begin{equation}
 \det{'}(\mathcal{L}) = \frac{2^{np-1}\ell^{\beta+1}}{p^{n-1}} \prod_{j=1}^{(n-1)/2} \left( p-\sum_{k=1}^p \cos \alpha_{k,j}  \right)^2 \ .
\end{equation}


\section{Spanning Trees} \label{sec: spanning trees}

The relation between the spectral determinant and number of spanning trees was studied in \cite{HWST}.
\begin{theorem}\label{thm: EST}
	Given an equilateral graph $\Ge$ with $E$ edges of length $\ell$ and $V$ vertices,
	\begin{equation}
	 \# \mbox{ spanning trees }=
	T_\Ge =\frac{\prod_{v\in \mathcal{V}}d_v}{E2^E\ell^{\beta+1}} \det{'}(\mathcal{L}) \ ,
	\end{equation}
	where $\beta=E-V+1$ is the first Betti number of the graph.
\end{theorem}
For a generic quantum graph $\Gamma$ with edge lengths in $[\ell, \ell+\epsilon]$, it was shown that the number of spanning trees is the closest integer to $T_\Gamma$ if $\epsilon$ is sufficiently small; see theorem \ref{thm:generic spanning trees}.  We now apply the results for the spectral determinants of almost equilateral graphs to estimate how fast $T_\Gamma$ deviates from the number of spanning trees.  

Given an almost equilateral quantum graph $\Gae$ with $E-1$ edges of length $\ell$ and one edge of length $\ell + \epsilon$, from equation (\ref{eq: spec det formula}),
\begin{equation}\label{eq: T formula with R}
T_{\Gae} = \frac{\prod_{v \in \mathcal{V}} d_v}{E2^E\ell^{\beta+1}}\det{'}(\mathcal{L^\prime}) = \frac{\ell_{tot}\prod_{e\in\mathcal{E}} \ell_e}{EV\ell^{\beta+1}}\det{'}(\mRt) \ ,
\end{equation} 
where
\begin{equation}\label{eq: coefficient l}
\ell_{tot}\prod_{e\in \mathcal{E}} \ell_e =  E\ell^{E+1} + \epsilon\ell^E(E+1) + \mathcal{O}(\epsilon^2) \ .
\end{equation}

By theorem \ref{thm:generic spanning trees}, we expect $T_{\Gae}$ to be close to the number of spanning trees when $\epsilon$ is small.  We will first look at an explicit example to demonstrate that $T_{\Gae}$ deviates from the  number of spanning trees when $\epsilon$ is large. Consider the equilateral complete graph $K_4$ (which has 4 vertices and 6 edges) with $\ell = 1$. In this case, the matrix $\mR$ has a simple eigenvalue of $0$ and an eigenvalue of $4$ (with multiplicity $3$), so $\det{'}(\mR) = 64$. By theorem \ref{thm: EST}, the number of spanning trees is
\begin{equation}
 \# \mbox{ spanning trees }=
	T_\Ge =\frac{\prod_{v\in \mathcal{V}}d_v}{E2^E\ell^{\beta+1}} \det{'}(\mathcal{L}) = \frac{\ell_{tot}\prod_{e\in\mathcal{E}} \ell_e}{EV\ell^{\beta+1}}\det{'}(\mR) = 16 \ .
\end{equation}

Now suppose one edge has length $2$, i.e. $\epsilon=1$. Without loss of generality, we take this to be the edge connecting the first and second vertices. Then,
\begin{equation}
\mRt = \begin{bmatrix}
2.5 & -0.5 & -1 & -1\\
-0.5 & 2.5 & -1 & -1\\
-1 & -1 & 3 & -1\\
-1 & -1 & -1 & 3
\end{bmatrix}
\end{equation}
which has eigenvalues of $0$, $3$, and $4$ (with multiplicity $2$), so $\det{'}(\mRt) = 48$. Therefore, 
\begin{equation}
T_{\Gae} = \frac{\ell_{tot}\prod_{e\in\mathcal{E}} \ell_e}{EV\ell^{\beta+1}}\det{'}(\mRt) = 28
\end{equation}
which is clearly not close to the number of spanning trees of $K_4$.

\subsection{Complete Graphs and Complete Bipartite Graphs}
Using equation (\ref{eq: C SD}), we can see that for an almost equilateral complete graph $K_n$,
\begin{align}
T_{\Gae}
& =\frac{\prod_{v\in \mathcal{V}}d_v}{E2^E\ell^{\beta+1}} \det{'}(\mathcal{L}) + \epsilon\left(\frac{\prod_{v\in \mathcal{V}}d_v(\beta+1)}{E^22^E\ell^{\beta+2}} \det{'}(\mathcal{L})\right) + \mathcal{O}(\epsilon^2)\\
& = \# \mbox{ spanning trees } + \frac{\epsilon}{\ell}\frac{n^{n-3}(n^2-3n+4)}{n-1} + \mathcal{O}(\epsilon^2),
\end{align}
since $\det{'}(\mathcal{L}) = E2^En^{n-2}\ell^{\beta+1}(n-1)^{-n}$, $\prod_{v\in\mathcal{V}}d_v = (n-1)^n$, $E = n(n-1)/2$, and $\beta = E-n+1$ for a complete equilateral graph.

For a large graph, $T_{\Gae}\approx \# \mbox{ spanning trees } + (\epsilon/\ell) n^{n-2} $.  If we use this to estimate the window of edge lengths for which $T_{\Gamma}$ can be used to evaluate the number of spanning trees, we would expect $[T_{\Gamma}]$, the nearest integer to $T_{\Gamma}$, gives the number of spanning trees provided $(\epsilon/\ell) n^{n-2}\leq 1/2E$.  This requires that the edge lengths all lie in an interval $[\ell,\ell+\delta)$ for $\delta \lesssim  \ell n^{2-n}/2E\approx \ell n^{-n} $.


Now we consider an almost equilateral complete bipartite graph, $K_{m,n}$. Substituting \eqref{eq: CB SD} into equation \eqref{eq: T formula with R}, 
\begin{align}
T_{\Gae} 
&= \# \mbox{ spanning trees } + \frac{\epsilon }{\ell}\frac{(\beta+1)\prod_{v\in\mathcal{V}}d_v}{E^2} + \mathcal{O}(\epsilon^2) \\
&=\# \mbox{ spanning trees } + \frac{\epsilon }{\ell}m^{n-2}n^{m-2}(mn-m-n+2) + \mathcal{O}(\epsilon^2)\label{eq: bipartite tree}\ ,
\end{align}
because for a complete bipartite graph $E = mn$, $\beta = E - m - n + 1$, and $\prod_{v\in \mathcal{V}} d_v = n^m m^n$.

To use this result to estimate a window of edge lengths for which $[T_{\Gamma}]$ could be expected to count the spanning trees, consider a large graph $K_{n,n}$. Then $T_{\Gae} \approx \# \mbox{ spanning trees } +(\epsilon/\ell) n^{2n-2}$.  Consequently $[T_{\Gamma}]$ might be expected to provide the number of spanning trees for a metric graph with edge lengths in a window $[\ell, \ell+\delta )$ for $\delta \lesssim \ell n^{-2n}/2$.

\subsection{Circulant Graphs}

Consider an almost equilateral circulant graph $C_n(\va)$ with a prime number of vertices and whose maximum degeneracy of eigenvalues of $\mR$ is two. Substituting (\ref{eq:R' aecg}) in (\ref{eq: T formula with R}) and using (\ref{eq:detR circulant}),
\begin{align}
&T_{\Gae} 
= \# \mbox{ spanning trees } + \epsilon \frac{\ell^{n-2}}{n}\det{'}(\mR)\left(\frac{np+1}{np} + C\right) + \mathcal{O}(\epsilon^2) \\
&= \# \mbox{ spanning trees } + \frac{ \epsilon}{\ell} \frac{2^{n-1}}{n} \left(\frac{np+1}{np} + C\right)    \prod_{j=1}^{(n-1)/2} \left( p-\sum_{k=1}^p \cos \alpha_{k,j}  \right)^2
   + \mathcal{O}(\epsilon^2) \ , \label{eq: Circ spanning trees}
\end{align}
where the number of edges is $E=np$,  the perturbed edge is $(0,a_m)$, and 
\begin{equation}
C=\displaystyle\sum_{j=1}^{(n-1)/2} \frac{2\cos \alpha_{m,j}\left(\cos \alpha_{m,j} - 1\right)}{n\left( p-\sum_{k=1}^p\cos \alpha_{k,j} \right) }\  .
\end{equation}

While equation (\ref{eq: Circ spanning trees}) holds for generic circulant graphs of prime order where eigenvalues of the adjacency matrix have a maximum multiplicity of two,
to use this to estimate the size of a window of edge lengths for which $[T_\Gamma]$ can be expected to count the number of spanning trees, we still need to gauge the size of $C$.  In order to estimate $C$ we can first consider the typical size of $\sum_{k=1}^p \cos \alpha_{k,j}$.  For large values of $p$ (which requires large graphs) and most angles $\theta \in (0,\pi)$ we can expect,
\begin{equation}
p-\sum_{k=1}^p \cos (\theta a_k) \approx p \ ,
\end{equation}
as the sum will behave like a sum evaluated at uniformly distributed random angles.
Then for a large graph,
\begin{equation}
C\approx \displaystyle\sum_{j=1}^{(n-1)/2} \frac{2\cos \alpha_{m,j}\left(\cos \alpha_{m,j} - 1\right)}{np } \approx \frac{1}{2p} \ ,
\end{equation}
where we replace the sum with averages assuming $n$ is large.  Hence, for a typical large $2p$-regular circulant graph,
\begin{equation}
T_{\Gae} 
\approx \# \mbox{ spanning trees } + \left(\frac{ \epsilon}{\ell}\right) \frac{1}{n} \left(1 + \frac{1}{2p} \right)     \left( 2p \right)^{n-1}
\  .
\end{equation}
Therefore we might expect $[ T_\Gamma ] $ to evaluate the number of spanning trees for graphs where the edge lengths fall in a window $[\ell,\ell+\delta )$ for $\delta \lesssim \ell  (2p)^{-n} (1+1/2p)^{-1}$.  

\subsection{Summary}

In theorem \ref{thm:generic spanning trees}, the integer $[T_\Gamma]$ is the number of spanning trees of a metric graph provided the edge lengths lie in a window $[\ell,\ell+\delta)$ with 
	\begin{equation}
\delta< \frac{\ell}{V^V \, 2^{E+V}\sqrt{2EV}}  \ .
\end{equation}
Considering the families of almost equilateral graphs analyzed in this section, this window appears unnecessarily narrow.  For complete graphs $K_n$ a window of size $\delta<\ell n^{-n}$ could be expected to be sufficient, while for complete bipartite graphs $K_{n,n}$ a window of size $\delta < \ell n^{-2n}/2$ should suffice.  For typical $2p$-regular circulant graphs $C_n(\va)$, edge lengths constrained to a window of size $\delta \lesssim \ell (2p)^{-n} (1+1/2p)^{-1}$ could be expected to be sufficient for $[T_{\Gamma}]$ to measure the number of spanning trees.  All of these results would agree, for example, with a window of size of $\ell V^{-V}$.  This is still a tight constraint which requires all of the edge lengths to be close together on the scale of $\ell$.  However, this argument suggests an analog of Kirchoff's matrix tree theorem for quantum graphs applies to graphs with edge lengths that fall in a window that is much wider than that required in \cite{HWST}. 

Alternatively, rather than estimate the size of an interval in which all of the edge lengths should lie, we could instead conjecture that $[T_\Gamma]$ should be the number of spanning trees of the metric graph provided that
\begin{equation}
	\sum_{e\in \mathcal{E}} \left|\frac{1}{E} - \frac{\ell_e}{\ell_{\mathrm{tot}}} \right| < V^{-V} \ ,
\end{equation}
where $\ell_{\mathrm{tot}}=\sum_{e\in \mathcal{E}} \ell_e$ is the total length of the graph. This formulation is obtained by requiring that the average of the distance from each edge length to the mean edge length is less than $(\ell_{\mathrm{tot}}/E)V^{-V}$. 

Theorem \ref{thm:generic spanning trees}, which originally appeared in \cite{HWST}, provides a version of Kirchhoff's matrix tree theorem for quantum graphs by connecting the number of spanning trees to the spectral determinant of the quantum graph for quantum graphs where the edge lengths are very close together.  
While, as our counterexample shows, it is not the case that $[T_{\Gamma}]$ provides the number of spanning trees of any quantum graph, there are reasons to think that $T_\Gamma$ should be somewhat robust to changes in the edge lengths. 
Consequently a quantum version of Kirchhoff's matrix tree theorem should connect the spectral determinant of a quantum graph to the number of spanning trees for a wider class of quantum graphs.  The results presented here provide a first step in improving the quantum version of a matrix tree theorem by estimating a wider window on which such a result might be proved for all metric graphs.  Alternative approaches may establish such a connection for quantum graphs with a wide window of edge lengths but with some additional constraint on the graphs, for example on the size of the spectral gap of the Laplacian, which is the first positive eigenvalue.  

\bmhead{Acknowledgments}

TW would like to thank the Rose-Hulman leaves program and the Baylor mathematics department for their hospitality during her sabbatical, when this work was carried out. The authors would like to thank the anonymous referees for their helpful comments.


\begin{thebibliography}{99}

\bibitem{Aetal00}
Akkermans, E., Comet, A., Desbois, J., Montambaux, G. and Texier, C.:
\newblock Spectral determinant on quantum graphs.
\newblock Ann. Phys. (2000) 
\newblock \href{https://doi.org/10.1006/aphy.2000.6056}{https://doi.org/10.1006/aphy.2000.6056}

\bibitem{BL}
Band, R. and Guillaume, L.:
\newblock Quantum graphs which optimize the spectral gap.
\newblock Ann. Henri Poincar\'{e} (2017)
\newblock \href{https://doi.org/10.1007/s00023-017-0601-2}{https://doi.org/10.1007/s00023-017-0601-2}

\bibitem{BHJ12}
Band, R., Harrison, J. M. and Joyner, C. H.:
\newblock Finite pseudo orbit expansions for spectral quantities of quantum graphs.
\newblock J. Phys. A.: Math. Theor. (2012)
\newblock \href{https://doi.org/10.1088/1751-8113/45/32/325204}{https://doi.org/10.1088/1751-8113/45/32/325204}

\bibitem{B08}
Berkolaiko, G.: 
\newblock Two constructions of quantum graphs and two types of spectral statistics. 
\newblock in Analysis on Graphs and its Applications, eds. Exner, P., Keating, J., Kuchment, P., Sunada, T. and Teplyaev, A. \emph{Proceedings of Symposia in Pure Mathematics 77}.  Amer. Math. Soc., Providence, RI, (2008) 315--331

\bibitem{BKKM}
Berkolaiko, G., Kennedy, J. B., Karasov, P. and Mugnolo, D.:
\newblock Surgery principles for the spectral analysis of quantum graphs.
\newblock Trans. Amer. Math. Soc. (2019)
\newblock \href{https://doi.org/10.1090/tran/7864}{https://doi.org/10.1090/tran/7864}

\bibitem{BKbook}
Berkolaiko, G. and Kuchment, P.:
\newblock  Introduction to quantum graphs, vol. 186 of {\em Mathematical
	Surveys and Monographs}.
\newblock Amer. Math. Soc., Providence, RI (2013)

\bibitem{BLS}
Berkolaiko, G., Latushkin, Y. and Sukhtaiev, S.:
\newblock Limits of quantum graph operators with shrinking edges.
\newblock Adv. in Math. (2019)
\newblock \href{https://doi.org/10.1016/j.aim.2019.06.017}{https://doi.org/10.1016/j.aim.2019.06.017}

\bibitem{BSW02}
Berkolaiko, G., Schanz, H. and Whitney, R. S.:
\newblock Leading off-diagonal correction to the form factor of large graphs.
\newblock Phys. Rev. Lett. (2002)
\newblock \href{https://doi.org/10.1103/PhysRevLett.88.104101}{https://doi.org/10.1103/PhysRevLett.88.104101}

\bibitem{BSW03}
Berkolaiko, G., Schanz, H. and Whitney, R. S.:
\newblock Form factor for a family of quantum graphs: an expansion to third order.
\newblock J. Phys. A: Math. Gen.  (2003)
\newblock \href{https://doi.org/10.1088/0305-4470/36/31/303}{https://doi.org/10.1088/0305-4470/36/31/303}

\bibitem{BW10}
Berkolaiko, G. and Winn, B.:
\newblock Relationship between scattering matrix and spectrum of quantum graphs.
\newblock Trans. Amer. Math. Soc. (2010)
\newblock \href{https://doi.org/10.1090/S0002-9947-2010-04897-4}{https://doi.org/10.1090/S0002-9947-2010-04897-4}




\bibitem{Davis}
Davis, P. J.:
\newblock Circulant Matrices.
\newblock Wiley, New York, NY (1979)


\bibitem{D00}
Desbois, J.:
\newblock Spectral determinant of Schr\"odinger operators on graphs.
\newblock J. Phys. A: Math. Gen. (2000) 
\newblock \href{https://doi.org/10.1088/0305-4470/33/7/103}{https://doi.org/10.1088/0305-4470/33/7/103}

\bibitem{D01}
Desbois, J.:
\newblock Spectral determinant on graphs with generalized boundary conditions.
\newblock Eur. Phys. J. B (2001) 
\newblock \href{https://doi.org/10.1007/s100510170013}{https://doi.org/10.1007/s100510170013}

\bibitem{D86}
van Doorn, E. A.:
\newblock Connectivity of circulant digraphs.
\newblock J. Graph Theory (1986)
\newblock \href{https://doi.org/10.1002/jgt.3190100103}{https://doi.org/10.1002/jgt.3190100103}

\bibitem{Fr}
Friedlander, L.:
\newblock Determinant of the Schr\"{o}dinger operator on a metric graph.
\newblock In: Berkolaiko, G., Carson, R., Fulling, S., and Kuchment, P. (eds.)
\newblock Quantum Graphs and their applications, p. 151-160
\newblock Amer. Math. Soc., Providence, RI (2006)

\bibitem{H24}
Harrison, J. M.: 
\newblock Quantizing graphs, one way or two? 
\newblock Rev. Math. Phys. (2024) 
\newblock \href{https://doi.org/10.1142/S0129055X24600018}{https://doi.org/10.1142/S0129055X24600018}

\bibitem{HK11}
Harrison, J. M. and Kirsten K.:
\newblock Zeta functions of quantum graphs.
\newblock J. Phys. A: Math. Theor. (2011) 
\newblock \href{https://doi.org/10.1088/1751-8113/44/23/235301}{https://doi.org/10.1088/1751-8113/44/23/235301}

\bibitem{HKT12}
Harrison, J. M., Kirsten, K. and Texier, C.:
\newblock Spectral determinants and zeta functions of Schr\"odinger operators on metric graphs.
\newblock J. Phys. A: Math. Theor.  (2012) 
\newblock \href{https://doi.org/10.1088/1751-8113/45/12/125206}{https://doi.org/10.1088/1751-8113/45/12/125206}

\bibitem{HP24}
Harrison, J. M. and Pruss, C.:
\newblock Circulant graphs as an example of discrete quantum unique ergodicity.
\newblock (Preprint) arXiv:2411.09028

\bibitem{HWST}
Harrison, J. M. and Weyand, T.:
\newblock Can one hear the spanning trees of a quantum graph?
\newblock Lett. Math. Phys. (2023) 
\newblock \href{https://doi.org/10.1007/s11005-023-01656-0}{https://doi.org/10.1007/s11005-023-01656-0}

\bibitem{HWK16}
Harrison, J. M., Weyand, T. and Kirsten, K.:
\newblock Zeta functions of the Dirac operator on quantum graphs.
\newblock J. Math. Phys. (2016) 
\newblock \href{https://doi.org/10.1063/1.4964260}{https://doi.org/10.1063/1.4964260}

\bibitem{K47}
Kirchhoff, G.: 
\newblock Uber die Aufl\"osung der Gleichungen, auf welche man bei
der Untersuchung der linearen Vertheilung galvanischer Str\"ome gef\"uhrt wird.
\newblock Ann. Phys. Chem. (1847) 72:497-508

\bibitem{KS99}
Kostrykin, V. and Schrader, R.:
\newblock Kirchhoff's rule for quantum wires.
\newblock J. Phys. A.: Math. Gen. (1999)
\newblock \href{https://doi.org/10.1088/0305-4470/32/4/006}{https://doi.org/10.1088/0305-4470/32/4/006}


\bibitem{Kuc03}
Kuchment, P.:
\newblock Quantum graphs {I}. {S}ome basic structures.
\newblock Waves Random Media (2003) 
\newblock \href{https://doi.org/10.1088/0959-7174/14/1/014}{https://doi.org/10.1088/0959-7174/14/1/014}

\bibitem{MMP}
Mathews, J. and Walker, R. L.:
\newblock Mathematical Methods of Physics, Second Edition.
\newblock Addison-Wesley Publishing Company Inc., Redwood City, CA (1970)

\bibitem{PM97}
Pascaud, M. and Montambaux, G.:
\newblock Magnetization of mesoscopic disordered networks.
\newblock Europhys. Lett. (1997)
\newblock \href{https://doi.org/10.1209/epl/i1997-00155-0}{https://doi.org/10.1209/epl/i1997-00155-0}

\bibitem{PM99}
Pascaud, M. and Montambaux, G.:
\newblock Persistent currents on networks.
\newblock Phys. Rev. Lett. (1999)
\newblock \href{https://doi.org/10.1103/PhysRevLett.82.4512}{https://doi.org/10.1103/PhysRevLett.82.4512}

\bibitem{SS00}
Schanz, H. and Smilansky, U.: 
\newblock Spectral statistics for quantum graphs: periodic orbits and combinatorics. 
\newblock Phil. Mag. B 80 (2000)
\newblock
\href{https://doi.org/10.1080/13642810010000635}{https://doi.org/10.1080/13642810010000635}


\bibitem{T01}
Tanner, G.:
\newblock Unitary-stochastic matrix ensembles and spectral statistics.
\newblock J. Phys. A: Math. Gen. (2001) 
\newblock
\href{https://doi.org/10.1088/0305-4470/34/41/307}{https://doi.org/10.1088/0305-4470/34/41/307}


\bibitem{WGTR}
Waltner, D., Gnutzmann, S., Tanner, G. and Richter, K.:
\newblock Subdeterminant approach for pseudo-orbit expansions of spectral determinants in quantum maps and quantum graphs.
\newblock Phys. Rev. E (2013)
\newblock \href{https://doi.org/10.1103/PhysRevE.87.052919}{https://doi.org/10.1103/PhysRevE.87.052919}




\end{thebibliography}

\appendix

\section{Appendix: Eigenvalues of an Almost Equilateral Complete Graph}
\label{sec: app complete}


For an equilateral complete graph $K_n$, the matrix $\mathbf{R}$ has eigenvalues $0$ and $n/\ell$ (with multiplicity $n-1$). Without loss of generality, suppose that the length of the edge connecting vertices 1 and $n$ is increased from $\ell$ to $\ell + \epsilon$.  Notice that $\{\mathbf{u}_k\}_{k=1}^{n-1}$ with,
\begin{equation}
[\mathbf{u}_k]_j =\frac{1}{\sqrt{k(k+1)}} \begin{cases}
1, & \mbox{ if } j\leq k\\
-k, & \mbox{ if } j = k+1\\
0, & \mbox{ otherwise} 
\end{cases}
\end{equation}
is a set of $n-1$ orthonormal eigenvectors of $\mR$ corresponding to the eigenvalue $n/\ell$. Then, from equation \eqref{eq: new Q}, the  $(n-1) \times (n-1)$ matrix $\mQt$ has entries
\begin{equation}
[\mQt]_{j,k} = \mathbf{u}_j\cdot \mathbf{Q}\mathbf{u}_k = \begin{cases}
\frac{1}{\sqrt{j(j+1)k(k+1)}}, & \mbox{ if } j,k \neq n-1\\
\frac{n}{\sqrt{j(j+1)k(k+1)}}, & \mbox{ if either } j \mbox{ or } k = n-1\\
 \frac{n}{n-1}, & \mbox{ if } j=k=n-1.
\end{cases}
\end{equation}
Define $n-3$ vectors $\{\mathbf{v}_k\}_{k=2}^{n-2}$ via,
\begin{equation}
[\mathbf{v}_k]_j = \begin{cases}
-\frac{1}{\sqrt{k(k+1)/2}}, & \mbox{ if } j=1\\
1, & \mbox{ if } j=k\\
0, & \mbox{ otherwise} , 
\end{cases}
\end{equation}
and a vector $\mathbf{v}_{n-1}$ with,
\begin{equation}
[\mathbf{v}_{n-1}]_j = \begin{cases}
-\frac{1}{\sqrt{n(n-1)/2}}, & \mbox{ if } j=1\\
\frac{1}{n}, & \mbox{ if } j=n-1\\
0, & \mbox{ otherwise}.
\end{cases}
\end{equation}
One can verify that $\{\mathbf{v}_k\}_{k=2}^{n-1}$ are $n-2$ linearly independent eigenvectors of $\mQt$ with eigenvalue $0$. 
Similarly, define $\mathbf{v}$ via,
\begin{equation}
[\mathbf{v}]_j = \begin{cases}
\frac{1}{\sqrt{j(j+1)/2}}, & \mbox{ if } j\neq n-1\\
\frac{n}{\sqrt{n(n-1)/2}}, & \mbox{ if } j =n-1.
\end{cases}
\end{equation}
Then $\mathbf{v}$ is an eigenvector of $\mQt$ with eigenvalue $2$. 
Hence, applying proposition \ref{prop: degen}, the eigenvalue $n/\ell$ has degeneracy $n-2$ (up to first-order in $\epsilon$) and there is a new simple eigenvalue, 
\begin{displaymath}
\frac{n}{\ell} -\frac{2\epsilon}{\ell^2} + \mathcal{O}(\epsilon^2) \ .
\end{displaymath}

The remaining eigenvalue of $\mathbf{R}$ is $\lambda =0$ with normalized eigenvector $\mathbf{u}$ where, 
\begin{equation}
[\mathbf{u}]_j = \frac{1}{\sqrt{n}} \qquad \mbox{for } j = 1, 2, \ldots, n.
\end{equation} 
Hence $[\mQ \mathbf{u}]_j = 0$ for all $j=1, \ldots , n$ and $ \lambda^\prime = 0 + \mathcal{O}(\epsilon^2)$ by proposition \ref{prop: nonde}.

Combining these results,
\begin{align}\label{eq: aecR}
\det{'}(\mRt) &= \prod_{j=2}^n \lambda_j^\prime 
= \left(\frac{n}{\ell}\right)^{n-1} - \epsilon \frac{2n^{n-2}}{\ell^n} + \mathcal{O}(\epsilon^2) \ .
\end{align}
Therefore, from theorem \ref{thm: Fr}, the spectral determinant of the almost equilateral complete graph $K_n$ is,
\begin{align}
\det{'}(\mathcal{L}^\prime) 
&= \frac{E2^{E}n^{n-2}\ell^{\beta+1}}{(n-1)^n} \left(  1+\frac{\epsilon}{\ell} \left(1-\frac{2(n-2)}{n(n-1)}\right)\right) +\mathcal{O}(\epsilon^2) \\
&=  \frac{E2^{E}n^{n-2}\ell^{\beta+1}}{(n-1)^n} \left(  1+\frac{\epsilon(\beta+1)}{\ell E} \right) +\mathcal{O}(\epsilon^2) \ ,
\end{align}
where $E=n(n-1)/2$ and $\beta=E-n+1$, which is \eqref{eq: C SD} in section \ref{sec: transitive graphs}.

\section{Appendix: Eigenvalues of an Almost Equilateral Complete Bipartite Graph}
\label{sec: app bipartite}


Consider a complete bipartite quantum graph $K_{m,n}$ where each edge has length $\ell$. Let $A = \{1, 2, \ldots, m\}$ and $B = \{m+1, m+2, \ldots, m+n\}$ where all vertices of $A$ are adjacent to all vertices of $B$; figure \ref{fig: complete bipartite} shows the complete bipartite graph $K_{3,5}$.  For such a graph, from equation (\ref{eq: R}), we have,
\begin{equation}
[\mathbf{R}]_{u,v} = \begin{cases}
\ell^{-1}n, & u=v \in A\\
\ell^{-1}m, & u=v \in B\\
-\ell^{-1}, & u\sim v\, (\mbox{so } u\in A \mbox{ and } v \in B \mbox{ or vice versa})\\
0, & \mbox{otherwise} \ .
\end{cases}
\end{equation}
The eigenvalues of $\mathbf{R} = \ell^{-1}\mathbf{L}$  are $\ell^{-1}(m+n)$, $\ell^{-1}n$ (with multiplicity $m-1$), $\ell^{-1}m$ (with multiplicity $n-1$) and $0$. Therefore
\begin{equation}\label{eq: det CB}
\det{'}(\mathbf{R}) = \frac{(m+n)m^{n-1}n^{m-1}}{\ell^{m+n-1}} \ .
\end{equation} 

Without loss of generality, we perturb the length of the edge $(1,m+1)$ and apply propositions \ref{prop: nonde} and  \ref{prop: degen} to obtain the following proposition for the spectrum of $\mRt$, which is proven in the following  subsection. 

\begin{proposition}\label{thm: eigenvalues of R}
Let $\Ge = K_{m,n}$ be an equilateral complete bipartite graph with edge length $\ell$. Suppose a new graph $\Gae$ is formed by increasing the length of one edge to $\ell + \epsilon$. Up to first-order in $\epsilon$, the eigenvalues of the matrix $\mRt$  are $0$, $\ell^{-1}(m+n) - \epsilon \ell^{-2}\left(\frac{m+n}{mn}\right)$, $\ell^{-1}n$ (with multiplicity $m-2$), $\ell^{-1}m $ (with multiplicity $n-2$), $\ell^{-1}n - \epsilon \ell^{-2} \left(\frac{m-1}{m}\right)$, and $\ell^{-1}m - \epsilon \ell^{-2} \left( \frac{n-1}{n}\right)$.
\end{proposition}

\subsection{Eigenvalues of $\mR'$}
\label{sec: bnondeg}

First, consider the simple eigenvalue $\lambda=\ell^{-1} (m+n)$ with normalized eigenvector $\mathbf{v}$ where,
\begin{equation}
[\mathbf{v}]_j = \frac{1}{\sqrt{mn(m+n)}}\begin{cases}
n, & \mbox{ if } 1\leq j \leq m\\
-m, & \mbox{ if } m+1 \leq j \leq m+n.\\
\end{cases} 
\end{equation}
Applying proposition \ref{prop: nonde}, we see that
\begin{align}
\lambda^\prime & = \frac{m+n}{\ell} -\frac{\epsilon}{\ell^2} \mathbf{v}\cdot\mathbf{Q}\mathbf{v} + \mathcal{O}(\epsilon^2) \\
&=  \frac{m+n}{\ell} - \frac{\epsilon}{\ell^2}\frac{m+n}{mn} + \mathcal{O}(\epsilon^2) \ ,
\end{align}
is an eigenvalue of $\mR'$.

Similarly, the eigenvalue $\lambda=0$ has normalized eigenvector $\mathbf{u}$ with,
\begin{equation}
[\mathbf{u}]_j = \frac{1}{\sqrt{m+n}} \qquad \mbox{for } j=1, 2, \ldots, m+n.
\end{equation}
Since $[\mathbf{Q}\mathbf{u}]_j = 0$,  we see that $\lambda^\prime = 0 + \mathcal{O}(\epsilon^2)$ is another eigenvalue of $\mR'$ using proposition \ref{prop: nonde}.


We now consider the eigenvalue $\ell^{-1}n$ with $m-1$ orthonormal eigenvectors $\{\mathbf{v}_k\}_{k=1}^{m-1}$ where,
\begin{equation}
[\mathbf{v}_k]_j = \frac{1}{\sqrt{k(k+1)}}\begin{cases}
1, & \mbox{for } j\leq k\\
-k, & \mbox{for } j = k+1\\
0, & \mbox{for } j > k+1.
\end{cases} 
\end{equation}
Then,
\begin{equation}\label{eq: tilde q}
[\mQt]_{j,k} = \mathbf{v}_j\cdot \mathbf{Q}\mathbf{v}_k = \frac{1}{\sqrt{j(j+1)k(k+1)}} \ .
\end{equation}
\begin{lemma}\label{thm: eigenvalues Q}
	The eigenvalues of $\mQt$ are $0$ (with multiplicity $m-2$) and $\frac{m-1}{m}$.
\end{lemma}
\begin{proof}
	Consider the vectors $\{\mathbf{u}_k\}_{k=2}^{m-1}$ of length $m-1$ such that
	\begin{equation}
	[\mathbf{u}_k]_j = \begin{cases}
	-\frac{1}{\sqrt{k(k+1)/2}}, & j=1\\
	1, & j=k\\
	0, &\mbox{otherwise}.
	\end{cases}
	\end{equation}
	Multiplying $\mathbf{u}_k$ by row $j$ of $\mQt$, we have
	\begin{equation}
	[\mQt\mathbf{u}_k]_j 
	=-\frac{1}{\sqrt{j(j+1)k(k+1)}} + \frac{1}{\sqrt{j(j+1)k(k+1)}} = 0 \ .
	\end{equation}
	Hence zero is an eigenvalue of $\mQt$ with multiplicity $m-2$.
	
	Similarly, define $\mathbf{w}$ such that
	\begin{equation}
	[\mathbf{w}]_j = \frac{1}{\sqrt{j(j+1)/2}} \qquad \mbox{for } j=1, \ldots, m-1.
	\end{equation}
	Then,
	\begin{align}
	[\mQt\mathbf{w}]_j &
	 = \frac{1}{\sqrt{j(j+1)/2}}\sum_{k=1}^{m-1} \frac{1}{k(k+1)}
	= \left(1-\frac{1}{m}\right)[\mathbf{w}]_j \ .
	\end{align}
	Hence, $1-1/m$ is the remaining eigenvalue of  $\mQt$.
\end{proof}


Similarly, $\ell^{-1}m$ is an eigenvalue of $\mR$ with multiplicity $n-1$ and the eigenvalues of a corresponding matrix $\mQt$ are 
$0$ (with multiplicity $n-2$) and $1-1/n$.
Combining the results in section \ref{sec: bnondeg}  with proposition  \ref{prop: degen} produces proposition \ref{thm: eigenvalues of R}.

\subsection{Spectral Determinant}
\label{sec: SD}

The spectral determinant of $\mRt$ is the product of its non-zero eigenvalues.  From proposition \ref{thm: eigenvalues of R}, for an almost equilateral complete bipartite graph,
\begin{align}\label{eq: det R}
\det{'}(\mRt) 
&=\det{'}(\mR) - \frac{\epsilon}{\ell^{m+n}} n^{m-2}m^{n-2}(m+n)(m+n-1)+ \mathcal{O}(\epsilon^2) \ .
\end{align}
Therefore, using equation (\ref{eq: det CB}) and theorem \ref{thm: Fr} with $E=mn$ and $V=m+n$, the spectral determinant of the almost equilateral complete bipartite graph $K_{m,n}$ is,
\begin{align}
\det{'}(\mathcal{L}^\prime) 
 & = 2^{E}\ell^{\beta+1} \left(1 + \frac{\epsilon}{\ell} \left(1-\frac{V-2}{E}\right)\right)+\mathcal{O}(\epsilon^2) \  \label{eq: aeL cb}\\
 & =2^E \ell^{\beta+1} \left(  1+\frac{\epsilon(\beta+1)}{\ell E} \right) +\mathcal{O}(\epsilon^2) \ ,
\end{align}
which is \eqref{eq: CB SD} in section \ref{sec: transitive graphs}.

As an explicit example, a \textit{star graph} with $n+1$ vertices and $E=V-1=n$ edges is the complete bipartite graph $K_{n,1}$. Suppose this is an equilateral graph with edge length $\ell$. The eigenvalues of $\mR$ are $\ell^{-1}$ (with multiplicity $n-1$), $0$ and $\ell^{-1}(n+1)$. 
Then, using \eqref{eq: det R} and perturbing one edge,
\begin{equation}\label{eq: detstar}
\det{'}(\mRt) 
= \frac{V}{\ell^{V-1}} - \epsilon\frac{V}{\ell^V} + \mathcal{O}(\epsilon^2) \ .
\end{equation}
From (\ref{eq: aeL cb}), 
\begin{equation}
\det{'}(\mathcal{L}^\prime) = 2^{E}\ell\left(1 +\frac{\epsilon}{\ell E}\right) +  \mathcal{O}(\epsilon^2) \ .
\end{equation}

Since a star graph is a tree, it has only one spanning tree. Suppose that  $\Gae$ is an almost equilateral star graph with $n+1$ vertices. For a star graph, using \eqref{eq: T formula with R} and \eqref{eq: detstar}, 
\begin{equation}
T_{\Gae} 
= 1 + \frac{\epsilon}{\ell n} + \mathcal{O}(\epsilon^2) \ ,
\end{equation}
which agrees with  \eqref{eq: bipartite tree} when $m=1$.
Consequently, for metric star graphs, we might expect $[T_{\Gamma}] =1 $ provided the edge lengths fall in a window $[\ell, \ell+\delta)$ for $\delta \lesssim \ell /2$, a much wider window than that found in theorem \ref{thm:generic spanning trees}.

\end{document}